\documentclass[11pt]{amsart}
  \usepackage{geometry}                
  \usepackage{graphicx}
 \usepackage{epstopdf}
 \usepackage{mydefs}
\usepackage{algorithm}
\usepackage{algorithmic}
 \usepackage{fullpage}
 \usepackage[foot]{amsaddr}

\usepackage{url}

\def\calL{{\mathcal{L}}}
\def\calU{{\mathcal{U}}}
\def\calM{{\mathcal{M}}}
\def\calT{{\mathcal{T}}}
\def\calP{{\mathcal{P}}}
\def\edges{{\mathcal{E}}}
\newcommand{\alg}[1]{\textbf{#1}}   

\def\Pro{{\mathbb{P}}}
\def\Ex{{\mathbb{E}}}

\title{Distributed Connectivity of Wireless Networks}
\author{Magn\'us M. Halld\'orsson and Pradipta Mitra}
\address{ICE-TCS, School of Computer Science\\
 Reykjavik University\\
 101 Reykjavik, Iceland\\}
\email{mmh@ru.is, ppmitra@gmail.com}

\begin{document}

\begin{abstract}
We consider the problem of constructing a communication infrastructure
from scratch, for a collection of identical wireless
nodes. Combinatorially, this means a) finding a set of links that form a
strongly connected spanning graph on a set of $n$ points in the plane, and b) scheduling it efficiently in the SINR model of interference.
The nodes must converge on a solution in a distributed manner, having
no means of communication beyond the sole wireless channel. 

We give distributed connectivity algorithms that run in time
$O(poly(\log \Delta, \log n))$, where $\Delta$ is the ratio between
the longest and shortest distances among nodes.  Given that algorithm
without prior knowledge of the instance are essentially limited to
using uniform power, this is close to best possible. Our primary aim,
however, is to find efficient structures, measured in the number of
slots used in the final schedule of the links.
Our main result is algorithms that
match the efficiency of centralized solutions. 
Specifically, the networks can be
scheduled in $O(\log n)$ slots using (arbitrary) power control, and in
$O(\log n (\log\log \Delta + \log n))$ slots using a simple oblivious power scheme. Additionally, the networks have the desirable properties that
the latency of a converge-cast and of any node-to-node communication is
optimal $O(\log n)$ time.
%
%
\end{abstract}

\maketitle

\section{Introduction}
\label{sec:intro}

We consider the problem of constructing a communication infrastructure
from scratch, for a collection of identical wireless
nodes. Combinatorially, this means finding a set of links that form a
strongly connected spanning graph on a set of points in the plane,
and scheduling it efficiently in the SINR model of interference.
The nodes must converge on a solution in a distributed manner, having
no means of communication beyond the sole wireless channel. The issue
is how quickly and how well: the time it takes to form the structure
and the efficiency of the final schedule produced. 


The importance of creating a connected structure spanning a set of wireless nodes can hardly be overstated. 
This may underlie a ``multi-hop''  wireless network,
where any two nodes can communicate through path(s) specified by such a structure. In an ad-hoc network, such a structure may provide the
underlying backbone for synchronized operation of the network. In a wireless sensor network, the structure can double as an information
aggregation mechanism. 

The efficiency of a structure is closely intertwined with the issue of interference, the distinguishing feature of wireless communication.
Interference implies that only a limited number of transmissions can
be successful simultaneously; this number depending on spatial
distribution of the links, power settings, etc. We adopt the SINR (or
physical) model of interference, that has been shown both
theoretically and experimentally to be a more faithful
representation of reality than many of the traditional graph-based models \cite{MaheshwariJD08,Moscibroda2006Protocol}. 

Achieving an efficient schedule involves deciding power levels for the links -- which may either be fully instance-dependent (``arbitrary''), or be chosen in an ``oblivious'' manner, depending only on the length of each link.
Recent centralized results show that it is possible to connect any
link set using $O(\log n)$ slots \cite{SODA12}, whereas the use of
oblivious power is bound to involve a factor of $\log\log \Delta$
\cite{us:esa09full,FKRV09,SODA12}, where $\Delta$ is the ratio
between shortest and longest distance in the network.

Achieving connectivity is a distributed problem \emph{par excellence}.
Distributed algorithms often assume ``free'' local
communication. In contrast, since the purpose in this paper is to build a communication
infrastructure from scratch, we assume that 
the only mode of communication allowed is transmission in the single
wireless channel, which succeeds if the required
signal-to-interference-and-noise ratio is achieved.
We also do not assume a \emph{carrier sensing} primitive (see, e.g.,
\cite{ScheidelerRS08}) that allows
nodes to estimate the amount of activity on the channel.

Given that the nodes have no information about distances to nearby
nodes, they are in effect limited to using a pre-defined fixed power initially.
It is known that usage of such a simple power scheme can necessarily require a linear number of slots to connect the nodes \cite{MoWa06}.
A more
refined bound is $\log \Delta$, where $\Delta$ is the ratio between
maximum to the minimum distance among the nodes.
We provide a distributed algorithm that forms a (initial) connected network in time 
$O(\log \Delta \cdot \log n)$, which is probably close to the best possible.

The quality or efficiency of the \emph{final} structure is another story.
Once the initial (and possibly inefficient) network is formed, 
we are interested in
retooling the network, still in a distributed fashion, but using the
existing network, to find improved connectivity structures. We provide two 
approaches to this. First, we show that the initial network has nice geometric
properties that allows us to use (distributed) power control to make it much more efficient. Second, we propose a more sophisticated approach --- instead of simply changing the power settings of the links of initial network, we leverage the initial tree to construct \emph{new} set of links (and their power settings) that can be scheduled even more efficiently, while still achieving connectivity.
This suggest a novel interplay
between different layers --- a network layer (i.e., the initial tree)  that goes back and retools both itself (choosing new links) and the MAC layer (changing power settings and schedules). 

The challenge raised in this paper can then be stated as follows:
\begin{quote}
Is there a distributed algorithm, running in time 
$O(poly(\log \Delta, \log n)$, 
that results in a nearly optimal strongly connected structure in the SINR model?
\end{quote}

We answer this question affirmatively, giving algorithms that match
the best upper bounds known for centralized algorithms. This holds
both for oblivious power assignments as well as when allowing
arbitrary power assignments. In particular, using arbitrary power, we
find and schedule a bidirectional tree in $O(\log n)$ slots that has
the property that both aggregation computation and any pairwise
communication can be achieved in optimal logarithmic time.

The rest of the paper is organized as follows. We introduce the model
and key definitions in Sec.~\ref{sec:model},
and discuss related results in Sec.~\ref{sec:related}.
Our results are described in Sec.~\ref{sec:results}. 
Section \ref{sec:technical} contains technical definitions and
clarifications that are essential for the analysis but not needed to
understand the results. The algorithm for the initial network construction is given and analyzed in
Sec.~\ref{sec:main}. Our two approaches to finding extremely efficient schedules are presented in Sec.~\ref{sec:sparsitypc1} 
and Sec.~\ref{sec:optimal}, respectively.
Several proofs and construction details have been deferred to appendices.

\section{Related Work}
\label{sec:related}

Connectivity was the first problem studied from a worst-case
perspective in the SINR model. In a seminal paper, Moscibroda and Wattenhofer \cite{MoWa06}
formalized the problem and proposed an algorithm that connects
arbitrary set of $n$ points in $O(\log^4 n)$ slots. This was improved
to $O(\log^3 n)$ \cite{moscibroda06b}, $O(\log^2 n)$ 
\cite{Moscibroda07}, and recently to $O(\log n)$ \cite{SODA12}.  
All these works deploy centralized algorithms.
No non-trivial lower bound is known.
Somewhat orthogonally, a large body of work exists on randomly
deployed wireless networks, starting with the influential work by
Gupta and Kumar \cite{kumar00}. Work in this setting for connectivity includes
\cite{Bettstetter01012004}, which studied the probability of
there existing a path between two nodes in a randomly deployed network.
In \cite{RodopluMinEnergy}, minimum energy connectivity structures is
studied for randomly deployed networks, but interference is
essentially ignored.

Distributed connectivity of wireless networks has also been the subject of research. In \cite{Zavlanos07distributedconnectivity}, connectivity in mobile networks was studied from
a graph-theoretic perspective with no explicit interference model. Indeed, connectivity maintenance problem has been well studied in control theory and 
robotics \cite{Zavlanos07distributedconnectivity,Meng05distributedconnectivity,Dimarogonas10}, but with different underlying assumptions, typically without the use of the SINR interference model. Sensor connectivity has also been studied \cite{Huang:2007:DPE:1210669.1210674} without reference
to any particular interference model. In \cite{ramanathandistributedconnectivity}, a heuristic was proposed for connectivity maintenance in
multi-hop wireless networks. A more rigorous study was done in
\cite{Wattenhofer01distributedtopology} but with the assumption of an
underlying MAC layer that resolves interference problems.


Two fundamental problems that deal with a given set of links
relate to this work.
\emph{Capacity}: find the largest feasible subset of links, and
\emph{Scheduling}: partition the link set into the fewest number of
feasible sets. 
For the former, constant-factor algorithms were given for uniform
power \cite{GHWW09,HW09}, mean and linear power (and most other
oblivious power assignments) \cite{SODA11}, and power control
\cite{KesselheimSoda11}. These imply a logarithmic factor for the
corresponding scheduling problems. Distributed algorithm was given
for \emph{Scheduling} with oblivious power \cite{KV10} and shown to achieve
$O(\log n)$-approximation \cite{icalp11}. 

Distributed algorithms have also been given for local broadcasting 
\cite{Goussevskaia2008Local} and dominating set \cite{ScheidelerRS08}
in the SINR model. Both of these problem are, however, local in nature.

The Minimum-Latency Aggregation Scheduling problem is closely related
to connectivity, where the latency for transmitting messages to a sink
is to be minimized. A large literature is known, but the first
worst-case analysis in the SINR model was given in
\cite{DBLP:conf/mswim/LiHWL10}, with a $O(\log^3 n)$ bound on the
schedule length by a centralized algorithm and $O(\log \Delta)$ by a
distributed algorithm.  The centralized bound was improved to optimal
$O(\log n)$ in \cite{SODA12}.

\section{Model and Preliminaries}
\label{sec:model}

Given is a set $P$ of $n$ wireless nodes located at points on the plane. Without loss of generality, assume that the minimum distance 
between any two points is $1$.
The nodes have synchronized clocks, and start running the distributed
algorithm simultaneously using slotted time. Each node knows its
location and has a globally unique ID. A single message is large
enough to contain the ID and the location of a node. A receiver of a
message thus always knows its distance from the sender and can
identify the sender uniquely. 

A \emph{link} is a directed edge between two nodes, indicating a
transmission from the first node (the sender) to the second (the
receiver). A link between $u$ and $v$ is denoted by $(u, v)$;
$\ell$ will also be used to indicate a generic link. A link set $L$
naturally induces a set of senders $S(L)$ and a set of receivers
$R(L)$.  
The link $(y, x)$ is known as the \emph{dual} of link $(x, y)$,
following \cite{KV10}. 
A link set $X$ is a dual of set $Y$ if $X$ consists of the duals of
the links in $Y$.
The degree of a node $u$ in a linkset $L$ is the number of links
incident on $u$ in $L$. The \emph{distance} between two nodes $u$ and $v$ is denoted by $d(u, v)$ (this is also the \emph{length} of the link $(u, v)$).
 Let $\Delta$ denote the maximum length of a possible link.  A
\emph{length class} refers to a set of links whose lengths differ by a
factor of at most 2.

In the SINR model of interference, a non-transmitting node $v$
successfully receives a message transmitted by node $u$ if,
\begin{equation}
\frac{P_u/d(u, v)^{\alpha}}{N + \sum_{w \in S \setminus \{u\}} P_w/d(w, v)^{\alpha}} \geq \beta\ ,
\label{gen_sinr}
\end{equation}
where $N$ is the ambient \emph{noise}, $\beta$ is the required SINR level,
$\alpha > 2$ is the so-called path loss constant, $P_w$ is the \emph{power} used
by node $w$, and $S$ is  the set of senders transmitting simultaneously.
A set $L$ of links is \emph{feasible} if the above constraint holds for all 
$v \in R(L)$ where $S = S(L)$. We do not impose any limit on the power a node can use.

The goal is to identify a set $\calT$ of links that both \emph{strongly
  connects} the wireless nodes and can be scheduled efficiently (i.e.,
can be partitioned into few feasible sets). Additionally, we seek low
latency constructions. 

A \emph{converge-cast}
tree is a directed rooted spanning tree where all links are oriented towards
the root (i.e., for each link, the receiver is a parent of the
sender).  An \emph{aggregation tree} is a converge-cast tree along
with a schedule of the links that has the property that each link $(x,
y)$ in the tree is scheduled after all links involving descendants of
$x$.  A \emph{dissemination tree} is the opposite: a broadcast tree
(spanning arborescence) with links oriented away from the root, with the
opposite property for the schedule. 
In both cases, the scheduling order follows link directions and paths in the trees.

\begin{defn}
A \textbf{bi-tree} is an aggregation
tree with a complementary dissemination tree, using the same links in
the opposite direction and same schedule in opposite order.
\end{defn}

Note that with a bi-tree, any node-node communication can be achieved
within time equal to the length of the schedule. The same holds for
computing an aggregation or a broadcast.

The following power assignments are of interest. An oblivious power
assignment is one where power assigned to a sender $u$ is a (simple)
function of $d(u, v)$, where $v$ is the intended receiver. The
oblivious assignment we are most interested in is the ``mean power''
assignment $\calM$ where $P^\calM_u = d(u, v)^{\alpha/2}$. 
We  also use uniform power $\calU$ that assigns the same power to all transmitting nodes, and the ``linear power'' assignment $\calL$ where
$P^\calL_u = d(u, v)^{\alpha}$.  
Note that a sender can transmit to
different receivers at different times, and may use different
powers. Finally, we also consider solutions achievable with
arbitrary power assignments, where the algorithm is free to use any
assignment.  We let $\Upsilon = O(\log\log \Delta + \log n)$ denote the
best ratio known for the cost of using oblivious power; namely, it is
known that for any  set of links, the ratio between the maximum size of feasible subset using arbitrary power vs.\ using mean power is at most $\Upsilon$ \cite{us:esa09full,SODA11}.

\section{Our Results}
\label{sec:results}

We give the first distributed algorithms with performance guarantees for connectivity problems in the SINR model. 
We first provide a basic algorithm: 

\begin{theorem}
There exists a distributed algorithm that computes a bi-tree $\calT$
in $O(\log \Delta \cdot \log n)$ slots.
\label{thm:1}
\end{theorem}

We can improve this solution by using
scheduling with non-uniform (but oblivious) power assignments.
Recall $\Upsilon = O(\log \log \Delta + \log n)$.

\begin{theorem}
The bi-tree $\calT$ can be re-scheduled in $O(\Upsilon \cdot \log^3
n)$ slots using mean power.
\label{thm:2}
\end{theorem}

We then 
intersperse the connectivity-building and the scheduling to get
solutions matching the best centralized solutions known. 

\begin{theorem}
There exists a distributed algorithm (building on the first one)
that finds and schedules a bi-tree in $O(\log n)$ slots (with arbitrary power), 
using time $O(\Upsilon \cdot \log \Delta \cdot \log n)$.
A variation finds and schedules a bi-tree in $O(\Upsilon \cdot \log
n)$ slots with mean power, using time $O(\Upsilon \log \Delta \cdot \log^2 n)$.
\label{thm:3}
\end{theorem}

In particular, the bi-tree property ensures that aggregation,
broadcast, and pairwise communication can all be achieved in optimal
$O(\log n)$ steps.

Technically, this work combines ingredients from numerous recent
works on the SINR model
\cite{us:esa09full,SODA11,KesselheimSoda11,KV10,icalp11,SODA12}. In addition,
we derive a number of properties, most of which deal
with the concept of \emph{affectance} in relation to connectivity
structures; intuitively, affectance measures the interference of one transmission on the reception of another transmission, relative to the signal strength of the latter.
We explicitly define a previously considered geometric
property of \emph{sparsity}, and show it to imply small average
affectance. We give novel algorithms for finding large feasible
subsets in such sparse link sets.  And, we introduce randomized
transmission strategies to estimate affectance in terms of transmission
successes.

\section{Technical Notes}
\label{sec:technical}

Our algorithms require the following knowledge about the instance:
The number of nodes, $n$, up to a polynomial factor; 
the minimum distance (assumed to be 1); and the maximum distance $\Delta$.
We do not treat $\Delta$ as a constant, although it is small in
many systems. Knowledge of $\Delta$ is mainly needed for stopping
criteria; it can be avoided by computing the size of the tree, if
precise knowledge of $n$ is available.

In describing our algorithm, we refer to some messages as
\emph{broadcasts} and some as \emph{acknowledgments}. In terms of if
and how these messages succeed, they are identical and work as
dictated by Eqn.~\ref{gen_sinr}. The difference lies in when these
messages are transmitted and what they contain. A broadcast refers
to an exploratory message sent to no node in particular, only containing
the sender's ID and location.  An acknowledgment is transmitted as a
response to a previous message (typically a broadcast) and 
contains IDs of both the sender (the acknowledger) and the 
initial broadcaster. Thus, receivers receiving an acknowledgment can
determine if it was addressed to them or not.


All our results are proved to be true ``with high probability'' 
(w.h.p., for short), where the term means that the relevant event 
occurs with probability $1 - \frac1{n^c}$, for some suitably large
$c$\footnote{This can be amplified to hold for \emph{any} $c$, by scaling up the constant factors.}. We frequently prove a lemma to hold, w.h.p., for a node $u$, or a
link $(u, v)$. It will always be clear that such a result can be
safely union bounded over all nodes, or all possible links, to derive
a high probability result for the whole algorithm.  The only case that
needs care is when we union bound over slots in the
algorithm. The number of slots in our first algorithm is a
function of $\log \Delta$, which can be arbitrarily larger than
$n$. Union bounding is still safe for the following reason.
The algorithm proceeds by considering links belonging to the same length class, and there can be at most $\log \Delta$ of such classes (thus the dependence on $\log \Delta$). However, since there are at most $n^2$ links in the network, only $n^2$ classes can actually be non-empty (in the full version, we provide a more refined $O(n)$ upper bound). During empty length classes, nothing happens with probability $1$ and thus the union bounding incurs no ``cost''.


\emph{Affectance. }
We use the notion of \emph{affectance}, introduced in
\cite{GHWW09,HW09} and refined in \cite{KV10} to the threshold-ed form
used here.  
The affectance $a^{\calP}_w(\ell)$ \emph{on} link $\ell=(u, v)$ \emph{from} a sender $w$,
with a given power assignment $\calP$,
is the interference of $w$ on $u$ relative to the power
received, or
\[
a^\calP_{w}(\ell)
  = \min\left\{1 + \epsilon, c(u, v) \frac{P_w}{P_u} \cdot \left(\frac{d(u, v)}{d(w, v)}\right)^\alpha\right\}\ ,
\]
where $\epsilon$ is some arbitrary fixed constant (say $0.1$), $c(u, v) = \beta/(1 - \beta N d(u, v)^\alpha/P_u)$ depends
only on the parameters of the link $\ell$.
We drop $\calP$ when clear from context. 
For a set $S$ of senders and a link $\ell$,
$a_S(\ell) = \sum_{w \in S} a_w(\ell)$. 

Using such notation, Eqn.~\ref{gen_sinr} can be rewritten as
$a_{S}(\ell) \leq 1$, which we adopt. 
When dealing with links $\ell = (u, v)$ and $\ell' = (u', v')$ we mean
$a_{\ell}(\ell')$ to mean $a_{u}(\ell')$.
Extending this to a link set $L$, we use the notation $a_L(\ell)$
to mean $a_S(\ell)$ where $S = S(L)$ are the senders in $L$. For two sets $X$ and $Y$, $a_X(Y)$ thus means $\sum_{\ell \in Y} a_{S(X)}(\ell)$.
From its definition, it
is clear that $c(u, v) \geq \beta$. We require that
$c(u, v) \leq 2\beta$, and point out how to achieve this
during the description of the algorithms.
This simply means that nodes always transmit
with power high enough for the intended (or potentially intended, in
case of a broadcast) links to comfortably succeed in the presence of
noise (but no other interference).

\section{Initial Tree Construction}
\label{sec:main}

The general template for the algorithm is as follows. At any given time,
a subset of the nodes is \emph{active}, with initially all nodes active and in the end only one node.
Links are formed between pairs of active nodes, 
by a node $u$ broadcasting, and another node $v$ acknowledging that
message in the next round. When such a communication succeeds, links $(u, v)$ and 
$(v, u)$ become part of the network and node $u$ becomes inactive (and forms no further links). The still active node
$v$ is $u$'s parent in the eventual aggregation tree. The link $(u, v)$
is then part of the aggregation tree and the link $(v, u)$ is part of the dissemination
tree. 

In what follows, $\lambda_1, \lambda_2 \ldots $, $\gamma_1, \gamma_2
\ldots $ are constants. 

The algorithm proceeds in $\lceil \log \Delta \rceil$ rounds, each containing $\lambda_1 \log n$ slot-pairs (a slot-pair is simply two consecutive slots). Each node $u$ maintains a link set
$L_u$ storing incoming and outgoing links along with a time stamp.
The final set $\calT$ is then simply $\cup_u L_u$.
In this initial tree construction, slots in the schedule of the links correspond simply to the time stamps.

At the beginning of each slot-pair in round $r$, each active node decides to be 
a \emph{broadcaster} with iid probability $p$ ($p \leq \frac12$ to be determined), and \emph{listener} otherwise. Then,
\begin{itemize}
  \item 
During the first slot, a broadcaster $u$ transmits a message and a listener $v$ listens for messages. 
\item 
During the second slot, a listener $v$ that received a message from $u$ such that $2^{r-1} \leq d(u, v) < 2^r$ in 
the previous slot does the following with iid probability $p$:
add the links $(u, v)$ and $(v, u)$ to $L_v$ with appropriate slot numbers and return an acknowledgment.
A broadcaster $u$ listens for acknowledgments during this slot, and on receiving one (say, from $v$) adds $(u, v)$
and $(v, u)$ to $L_u$, and becomes \emph{inactive}.
\end{itemize}

Note that a node only forms links with nodes at distance in the range
$[2^{r-1}, 2^r)$ during round $r$. Since each node knows this range it
can easily choose a power that ensures $c(u, v) \leq 2 \beta$ for
all $d(u, v) \in [2^{r-1}, 2^r)$.  Setting the power to $2 \beta N
2^{r \alpha}$ suffices.
We say that a link $(u, v)$ is \emph{successfully formed} between
nodes $u$ and $v$ during a slot-pair if all of the following happen:
a) the transmission $(u, v)$ is successful in the first slot, b) it is
successfully acknowledged in the second slot (i.e., the link $(v, u)$
successfully transmits), and c) both nodes store $(u, v)$ and $(v, u)$
in their set of links with the appropriate time stamps.  Note that
when this happens, $u$ becomes inactive, by the description of the
algorithm.
The sole link that is outgoing from a given node is also the last one
to be scheduled, thus ordering satisfies the leaf-to-root order of
aggregation trees.

\noindent \textbf{Remarks.} Two technical clarifications. First, note that
a listener $v$ can store a failed link, since it does not necessarily know 
whether an acknowledgment $(v, u)$ succeeded. However, this is not a problem, since:
a) Node $u$ remains active if the acknowledgment fails and connects itself later
to some node (or eventually becomes the root), b) Transmission of the link $(v, u)$ is not problematic for other links, since links transmitting in that slot \emph{did} succeed in the presence of that transmission. In any case, it is easy to efficiently ``clean up'' such stray links after the whole network is formed. Second,
as constructed, the dissemination tree has the opposite order of links in the schedule (links closer to the root are scheduled later, instead of earlier, as the definition calls for). This is also easily fixable after the network is formed by  a reversal process initiated by the root. We omit these details in this version.

\subsection{Analysis}

\noindent We first show that short links have a high probability of succeeding.
\begin{lemma}
Assume that at the beginning of round $r$, 
the minimum distance between active nodes is at least $2^{r-1}$.
Consider any slot-pair in the round and
active nodes $u$ and $v$ with 
$d(u, v) < 2^{r}$. Then, with 
probability at least $\frac14 p^2 (1-p)$, the link $(u, v)$ is
successfully formed in that slot-pair. Similarly, with probability at least $\frac14 p^2(1-p)$,
the link $(v, u)$ is successfully formed.
\label{smalllinksuccess3}
\end{lemma}
\begin{proof}
Let $\rho = 2^{r-1}$. Let $M_r$ be the set of currently active nodes
and let $\ell = (u, v)$.
Let $B_r$ be the set of broadcasters during the slot-pair.
First,  note that by the description of the algorithm
\[ \Pro(u \in B_r \text{ and } v \not \in B_r) = p (1 - p)\ . \]

For $t = 0, 1, \ldots $ define $C_t$ to be the ball around $v$ of radius $\rho (t + 1)$ and define the annulus $A_t$ as
 $A_0 = C_0$, $A_t = C_{t} \setminus C_{t-1}$ for $t \geq 1$. From
this it is easily computed that the area of $A_t$ is
\begin{equation}
\text{Area}(A_t) = \pi \rho^2 (2t + 1)
\label{eqn:trarea}
\end{equation}

Now, by the definition of $\rho$, balls of radius $\frac{\rho}{4}$
around any pair of points in $M_r$ do not intersect (since the minimum distance between active nodes is $\rho$). Combining this with
Eqn.~\ref{eqn:trarea}, we see that $A_t$ contains at most $16 (2t + 1) \le 48t$
nodes in $M_r$.

For $x \in M_r \cap A_0$, $a_x(\ell) \leq 1 + \epsilon$, simply by the definition of affectance. For $x \in M_r \cap A_t$ for $t \geq 1$, $d(x, v) \geq \rho \cdot t$
and thus $a_x(\ell) \leq c_{\ell} \frac{(2\rho)^{\alpha}}{(\rho \cdot t)^{\alpha}} \leq 2 \beta  \left(\frac{2}{t}\right)^{\alpha}$, where $c_{\ell} \equiv c(u, v) \leq 2 \beta$. Note that 
for any $x$, $\Pro(x \in B_r) = p$.

Thus,
\begin{align*}
\Ex(a_{B_r}(\ell)) & = \Ex(a_{B_r \cap A_0}(\ell)) + \sum_{t \geq 1}\Ex(a_{B_r \cap A_t}(\ell)) \\
 & \le 16(1+\epsilon) p + p 2 \beta  \sum_{t \geq 1} \left(\frac{2}{t}\right)^{\alpha} 48 t \\
 & \le 16(1+\epsilon) p + 96 p \beta 2^{\alpha} \frac1{\alpha - 2}\ ,
\end{align*}
using the bound $\zeta(x)=\sum_{n\ge 1} \frac{1}{n^x} \le \frac{1}{x-1}$ on the Riemann zeta function. 
Thus, for any $p \leq (64 (1 + 6 \beta 2^{\alpha} \frac1{\alpha - 2}))^{-1}$, 
we get that $\Ex(a_{B_r}(\ell)) \leq 1 /2$. By
Markov's inequality, $a_{B_r}(\ell) \leq 1$ with probability at least
$\frac{1}{2}$ (recall that this means that the link $\ell$ succeeds).
Thus,
\[ \Pro(a_{B_r}(\ell) \leq 1 \text{ and } u \in B_r \text{ and } v \not \in B_r) \geq \frac12 p (1 - p)\ ,\]
A similar argument proves that the probability of $\ell_r = (v, u)$ succeeding
is at least $\frac12 p$ and thus the link $(u, v)$ is formed with probability at least $\frac14 p^2 (1 - p)$.
The argument for the potential formation of link $(v, u)$ is identical.
\end{proof}

Now we can claim that,

\begin{lemma}
At the beginning of each round $r$, the distance between active nodes is at least $2^{r-1}$, w.h.p.
\label{distancelevels2}
\end{lemma}
\begin{proof} (Sketch.)
The claim is clearly true for round $1$ (since the minimum distance in the system
is $1$). Now inductively assume that it is true for round $r$. Consider any two nodes
$u, v$ that are active at the beginning of round $r+1$ with $d(u, v) < 2^{r+1}$.
Consider any slot-pair in which they are both active. By Lemma \ref{smalllinksuccess3}, the probability of both of them remaining active
after this slot pair is at most $1 - \frac14 p^2 (1 - p) \leq 1 - \frac18 p^2$.
Thus, the probability of both of them remaining active over $\lambda_1 \log n$
slot-pairs is $\leq (1 - \frac18 p^2)^{\lambda_1 \log n}$. Setting $\lambda_1 = \frac{80}{p^2}$, this probability can be upper bounded by $\frac1{n^{10}}$. This 
proves the Lemma (after union bounding).
\end{proof}

We can now prove the first main result.

\begin{proof}{[of Thm.~\ref{thm:1}]}
By Lemma \ref{distancelevels2} it is clear that within $O(\log \Delta)$ rounds, and
thus $O(\log \Delta \cdot \log n)$ slots, at most one active node remains (since the maximum distance among nodes is $\Delta$).
Since nodes only cease to be active by forming links with an active node, it
is also clear that exactly one node remains active. When nodes cease to be
active, they do so only by connecting in both directions to still-active nodes
(by the description of the algorithm). By induction, 
the whole network is then strongly connected to the single node
active at the end. This last active node is the root of both the aggregation and dissemination trees.
\end{proof}

We can also show that the network formed has low degree, where the degree 
of a node is its number $|L_u|$ of incident links.
%
\begin{theorem}
The probability of a link having degree $d$ is at most
$e^{\frac{- p^2 d}{8}}$. As a result, the maximum degree is
$O(\log n)$, w.h.p.
\label{seconddegbound}
\end{theorem}
\begin{proof}
Let $u$ be a node and consider any round $r$ and any slot-pair in the
round where $u$ is active. Suppose there is another active node $v$
with $d(u, v) < 2^r$. Then by Lemma \ref{smalllinksuccess3}, $u$
ceases to be active after this slot-pair, with probability
at least $\frac14 p^2 (1 - p) \geq \frac18 p^2$. Note that in slot-pairs
where no such $v$ exists, $u$ does not form a link. Thus, the degree of a node is upper bounded by the number of slot pairs where such a $v$ exists, and $u$ 
remains active after wards. The probability of there being $d$ such slot pairs is 
at most $(1 - \frac18 p^2)^d \leq e^{\frac{- p^2 d}{8}}$. 

Setting $d = \frac1{p^2} 80  \log n$ gives us the 
second part of the lemma.
\end{proof}

\section{Sparsity and Power Control}
\label{sec:sparsitypc1}

In this section, we show that the link set $\calT$ produced by the
algorithm of Sec.~\ref{sec:main}
can actually be scheduled in considerably fewer slots (in terms
of dependence on $\Delta$) using mean power, thus proving Thm.~\ref{thm:2}. This leads to an
algorithm to reschedule the same links with this improved power
assignment. The main idea is to show that the produced link set has
certain geometric properties that allows such improved scheduling.


\begin{defn}
A set $L$ of links is \textbf{$\psi$-sparse} if, for every closed ball $B$
in the plane,
   \[ B \cap L(8 \cdot rad(B)) \leq \psi\ , \]
where $rad(B)$ is
the radius of $B$, $L(d)$ is the set of links in $L$ of length at
least $d$, and $B \cap Q$ denotes the links in 
a set $Q$ with at least one endpoint in ball $B$.
\end{defn}

It was shown in \cite{SODA12} that the sparsity property (not
explicitly defined there) is connected to a property named
\emph{amenability} in \cite{SODA12}, which via an algorithm in
\cite{KesselheimSoda11} and results in \cite{SODA11} imply the
following:

\begin{theorem}[\cite{SODA12}]
Let $L$ be a $\psi$-sparse link set, for some $\psi \geq 1$. 
Then any $L' \subseteq L$ contains
a feasible subset of size $\Omega\left(\frac{|L'|}{\psi}\right)$. The set $L$ can be scheduled
in $O(\psi \log n)$ slots. Furthermore, any $L' \subseteq L$ contains
a subset of size $\Omega\left(\frac{|L'|}{\psi \cdot \Upsilon}\right)$ that is feasible under
mean power assignment.
The set $L$ can be scheduled in 
$O(\psi \cdot \Upsilon \cdot \log n)$ slots using mean power.
\label{sparsecapacity}
\end{theorem}
We provide a short overview of these ideas for reference in Appendix \ref{sec:amenabilityprimer}.

\noindent We now claim a sparsity result for the network $\calT$ formed by the algorithm.
\begin{lemma}
If $D$ is a disc of radius $\rho$ in the plane, then
the number of links in $\calT$ longer than $8\rho$ that have at least one endpoint in $D$ is $O(\log n)$, w.h.p.
\label{pointsinball}
\end{lemma}
\begin{proof}
Let $L = L(8 \cdot \rho) \cap D$.
We first claim that at most one node inside $D$ is incident to a link in $L$.
For contradiction, assume that there are two such nodes $u$ and $v$. Now, by the description of the
algorithm, links of length $8 \rho$ or higher can only be formed in rounds $\log \rho + 4$ or higher. Thus, both $u$ and
$v$ were active during round $\log \rho + 4$.
However, $d(u, v) \leq 2 \rho$ and thus by Lemma \ref{distancelevels2}, at the end of round $\log \rho + 2$,
at most one of them could remain active. This is a contradiction.
The proof of the Lemma is now complete by Thm.~\ref{seconddegbound}.
\end{proof}

By union bounding over all $\rho$ and all balls (by careful selection, there
are only polynomially many of them that are relevant), this implies:

\begin{theorem}
The set $\calT$ of links produced by the algorithm is $O(\log n)$-sparse.
\label{lognsparsity}
\end{theorem}

We now propose the following extension of the algorithm to 
schedule the links using significantly fewer slots. 
\begin{quote}
The sender of each link $\ell$ in $\calT$ sets its power to
mean power, $\ell^{\alpha/2}$. The links then use the distributed algorithm from
\cite{KV10} to compute a schedule of the links using this power
assignment. 
\end{quote}

We can now prove Thm.~\ref{thm:2}.
\begin{proof}
Thm.~\ref{sparsecapacity} and Thm.~\ref{lognsparsity} imply that $\calT$ 
can be scheduled in $O(\Upsilon \cdot \log^2 n)$ slots using mean power.
The distributed scheduling algorithm of \cite{KV10} produces a
$O(\log n)$-approximation \cite{icalp11}, giving the Theorem. (See Appendix \ref{sec:dualclarification} for a technical note on the approximation factor in \cite{icalp11}).
\end{proof}
The resulting schedule, however, does not necessarily satisfy the ordering
property of bi-trees.

\section{Matching Centralized Bounds}
\label{sec:optimal}
In this section, we prove Thm.~\ref{thm:3}. The difference with Sec.~\ref{sec:sparsitypc1} are threefold. 
First, we achieve more efficient final schedules. Second, unlike Sec.~\ref{sec:sparsitypc1}, we produce bi-trees. The third is a difference in approach. 
While the algorithm in Sec.~\ref{sec:sparsitypc1}  merely rescheduled the links in the original tree, in this section, we shall actually build a \emph{new} tree with superior properties, but will do so by using the original tree.

We use \alg{Init} to refer to the algorithm from Sec.~\ref{sec:main} that constructs the initial bi-tree. For any link set $L$ which is a subset of a directed rooted tree, we call a node $u$ a ``top level node'' with respect to $L$ if no link of form $(v, w)$ is in $L$ (i.e., the link between $v$ and its parent in the rooted tree, if such a link exists, is not present in $L$).

In what follows, we focus on forming the aggregation tree part for simplicity (constructing the dissemination tree portion of the bi-tree is essentially identical). 
The algorithmic framework is as follows.
\begin{algorithm}                      
\caption{TreeViaCapacity}          
\label{alg1}                           
\begin{algorithmic}[1]                    
      \STATE Set $i = 0$ and $P_i = P$ (the original input set). 
      \FOR{$i = 0, 1, 2 \ldots$ \textbf{until} $|P_i| = 1$}
        \STATE Construct (aggregation) tree $\calT$ on $P_i$ using \alg{Init}.
        \STATE Find a feasible subset $\calT' \subset \calT$
        \STATE Let $P_{i+1}$ be the set of top level nodes w.r.t.\ $\calT'$. 
      \ENDFOR
 \end{algorithmic}
 \label{alg1fig}
 \end{algorithm}

If $\calT'$ is large, then this process ends quickly.
\begin{theorem}
Assume that in each iteration, $\Ex(|\calT'|) = \delta |\calT|$ for
some $\delta > 0$. Then, the process ends after $O(\frac1{\delta}
\log n)$ iterations and the links produced form an aggregation tree
connecting the nodes in $O(\frac1{\delta} \log n)$ slots, w.h.p. 
\label{thm:logn}
\end{theorem}
\begin{proof}
First we show that:
\begin{claim}
 $\Ex(|P_{i+1}|) \leq (1 - \frac12 \delta) |P_{i}|$, for any $P_i$ such that $|P_i| \geq 2$.
 \end{claim}
 \begin{proof}
 Suffices to pro
Recall that $|\calT'| \geq \delta |\calT| = \delta (|P_i| -1)$.
Consider any link $(u, v) \in \calT'$. Clearly, this link rules out
$u$ as a top level node. Also, since $\calT$ is an aggregation tree,
there can be at most one outgoing link from each node $u$. Thus, 
$\Ex(|P_{i+1}|) \leq |P_i| -\Ex(|\calT'|) \leq |P_i| - \delta (|P_i| -1) \leq (1 - \frac12 \delta) |P_i|$ (for $|P_i| \geq 2$). 
 \end{proof}
 
This can be used to show that the process ends in $O(\frac1{\delta} \log
n)$ steps, w.h.p.

\begin{claim}
 $\Pro(|P_{t}| > 1) \leq \frac1{n^4}$ for $t = 10 \frac1{\delta} \log n$.
 \end{claim}
 \begin{proof}
 Since $P_i$ is non-increasing in $i$, for contradiction, condition on all $P_i \geq 2$ for $i \leq t$. Then we can apply the above Lemma to show that 
 
 \begin{align*}
 \Ex(|P_t|) \leq \left(1 - \frac12 \delta\right)^{10 \frac1{\delta} \log n} \frac1{n} \leq \frac1{n^4} \ ,
 \end{align*}
from which the claim follows by Markov's inequality.
 \end{proof}

By the definition of top level nodes, nodes not in $P_{i+1}$ are connected to some node in $P_i$ by a link. Thus, the final structure is clearly a converge-cast tree. The ordering on schedules is also guaranteed by the way the algorithm proceeds (it is easy to see that nodes can be involved in at most one link in a feasible set, thus the ordering is not violated within  $\calT'$).
 
Finally, since each iteration uses a single slot, the bound on
iterations immediately implies the bound on the number of slots in the
schedule. The theorem follows.
\end{proof}

To implement the above scheme, we need to show that $\calT'$ can always be found for a large enough $\delta$ to claim the results in Thm.~\ref{thm:3}.

We do this in two steps: in the first step a $O(1)$-sparse subset $\calT(M) \subseteq \calT$ is chosen, and in the second step a subset of $\calT(M)$ is chosen as $\calT'$. The first step is identical for mean power and arbitrary power case. The set $\calT(M)$ is defined in the following result, whose proof is in Appendix \ref{sec:missing}.
\begin{theorem}
Let $M$ be the set of nodes of degree at most $\rho = \frac{160}{p^2}$
in $\calT$, and let $\calT(M)$ be the links in $\calT$ induced by $M$.
Then, $\calT(M)$ is $O(1)$-sparse and $\Ex(|\calT(M)|) = \Omega(|\calT|)$.
\label{constantsparsity}
\end{theorem}



To actually compute $\calT(M)$ in a distributed fashion, 
 note that nodes can easily decide if they are in $M$ (by counting the number of links adjacent to them). One sweep through the  existing network $\calT$ is enough for each node to detect which of their links (if any) are in $\calT(M)$. 

Selecting $\calT'$ is also reasonably easy for mean power, but 
more involved for arbitrary powers. 
The following two subsections deal with these cases separately.
Note that we keep the original network around at all times,
which is useful for controlling the construction of the new one.
Running these networks in parallel can be achieved with
simple time-division multiplexing.

\subsection{Finding $\calT'$ with mean power}

Assume that $\calT(M)$ is known. It can be shown that the \emph{average affectance} in the linkset $\calT(M)$ (under mean power) is small, or $O(\Upsilon)$ (proof in Appendix \ref{sec:missing}).

\begin{lemma}
Affectance within $\calT(M)$ under mean power satisfies
$a^{\calM}_{\calT(M)}(\calT(M)) = \gamma_1 \Upsilon |\calT(M)|$, for some constant $\gamma_1$.
\label{averageaffectancemean}
\end{lemma}

Lemma \ref{averageaffectancemean} implies, after some basic manipulation, that
there exists $Q$ with $|Q| \geq \frac12 |\calT(M)|$, such that
$a^{\calM}_{\calT(M)}(\ell) \leq 2 \gamma_1 \Upsilon$ for all $\ell \in Q$.

The following \emph{sampling} mechanism produces a large feasible set 
in expectation (see \cite{FKRV09}):
Each link in $\calT(M)$ transmits with iid probability $\frac1{4
  \gamma_1 \Upsilon}$, with the successful links forming the set $\calT'$.
Since each transmitting link in $Q$ succeeds with probability $\geq \frac12$,
the expected size of $\calT'$ is at least $\frac1{2 \gamma_1 \Upsilon} |Q| = \Omega(\frac1{\Upsilon} |\calT(M)|)$. 
Combining this with Thm.~\ref{constantsparsity}, we get that
\begin{lemma}
$\Ex(|\calT'|) = \Omega(\frac1{\Upsilon} \Ex(|\calT(M)|)) = \Omega(\frac1{\Upsilon} |\calT|)$.
\end{lemma}

Thus, Thm.~\ref{thm:logn} can be invoked with $\delta =
\Omega(\frac1{\Upsilon})$, to obtain the second half of Thm.~\ref{thm:3}:

\begin{theorem}
  There exists a distributed algorithm that forms and schedules a
  bi-tree in
  $O(\Upsilon \cdot \log n)$ slots using mean power. This algorithm completes in time $O(\Upsilon \log \Delta \cdot \log^2 n)$.
  \label{mpfinal}
\end{theorem}
\begin{proof}
The performance of the final solution follows from
Thm.~\ref{thm:logn}, as mentioned above. Let us the bound the total
running time. The algorithm \alg{Init} needs to be invoked $O(\Upsilon
\cdot \log n)$ times, for a total cost of $O(\Upsilon \cdot \log
\Delta \cdot \log^2 n )$. After forming $\calT$ with each such
invocation, identifying $\calT(M)$ costs $O(\log \Delta \log n)$ (the
cost of $\calT$). Computing $\calT'$ is cheap since the sampling is
done in parallel. 
One technical aspect to note is that while the nodes choose $\calT'$,
they nodes need to know if their transmission succeeded; this can be
done without substantial loss of performance using an extra acknowledgment
slot, as we have seen before. 
The runtime bound of the theorem then follows.
\end{proof}
This theorem completes the proof of the second half of Thm.~\ref{thm:3}.

\subsection{Finding $\calT'$ with arbitrary power}
In this case, we want to find a large set $\calT'$, given $\calT(M)$, and then 
choose a power assignment making the set feasible. 

We start with the link selection step.
Leveraging the fact that our input instance $\calT(M)$ is sparse, 
we implement a distributed version of a centralized algorithm for
choosing such a set proposed in \cite{KesselheimSoda11}. 

The following algorithm was shown in \cite{KesselheimSoda11} to give constant factor approximation for finding the largest feasible subset of any given linkset: Given a linkset $R$, let the selected set be $L$, initially empty. Go through all links in ascending order of length (breaking ties arbitrarily). If the 
condition
\begin{equation}
a^{\calL}_{L}(\ell) + a^{\calU}_{\ell}(L) \leq \tau\ ,
\label{selctionalg}
\end{equation}
holds, for a constant $\tau$, then the link $\ell$ is added to $L$ (Eqn.~1 of \cite{KesselheimSoda11} can be seen to be essentially equivalent to the above equation). 

For simplicity, we assume in this abstract that that receivers can measure the SINR of a successful link (i.e., can measure if the link succeeded with a desired threshold $\tau$ or not). This assumption can be removed.

%
Assume the formation of $\calT$ using \alg{Init} required $R$ rounds.
Our selection algorithm \textbf{Distr-Cap} has the following outline.
\begin{quote}
\textbf{Distr-Cap} contains $R$ phases. In phase $i$, links in $\calT(M)$ that were formed in round $i$ of \alg{Init} decide whether or not to add themselves to the selected set $\calT'$.
\end{quote}

By the description of \alg{Init}, links formed in the same round belong to the same length class (also, links formed in a particular round are smaller than all links formed in later rounds).

For all $i$, phase $i$ of \textbf{Distr-Cap} consists of one slot-pair. Let $Q$ be the links participating in this phase (i.e., links formed during round $i$ of \alg{Init}).
During the first slot of the phase, the following happens:
\begin{enumerate}
\item All links $\ell$ in $\calT'$ (the set selected so far) transmit using linear power (i.e. $P_{\ell} = \ell^{\alpha}$).
\item Links in $Q$ transmit with iid probability $p$ (small constant) using linear power. 
\item Receivers in $Q$ record a success if they received a message across the link with SINR $\leq \tau/4$. Let $\tilde Q$ be the set of links that recorded success.
\end{enumerate}

During the second slot:
\begin{enumerate}
\item Links in $\calT'_d$ (dual of $\calT'$) transmit using linear power (i.e., the receivers of $\calT'$ transmit using linear power).
\item Links in $\tilde Q_d$ (dual of $\tilde Q$) transmit with iid
probability $\gamma_2^2 \cdot p$ for some $\gamma_2 < 1$, using linear
power.
\item Receivers in $\tilde Q_d$ record a success if they received a message across the link with SINR $\leq \frac{ \gamma_2 \cdot \tau}{4}$. 
\end{enumerate}

Thus, at the end of a second slot, a success is recorded at a sender
of a (original) link in $Q$, if the transmission succeeded in both
directions (the original link and the dual) with the required SINR
threshold. Let $Q^*$ be the set of links that succeeded.
The updated solution is then $\calT' \leftarrow \calT' \cup Q^*$,
which simply means that links add themselves to $\calT'$ if they
succeeded in both directions.

We now analyze this algorithm. The following sub-subsections show that
the selected solution is feasible and large (a constant factor
approximation to the largest feasible subset), respectively.

\subsubsection{$\calT'$ is feasible}

We now show that $\calT'$ satisfies Eqn.~\ref{selctionalg}.
It suffices to show that for all $\ell \in \calT'$, if $L \subseteq \calT'$ are the links no larger than $\ell$ then:
\[ a^{\calL}_{L}(\ell) + a^{\calU}_{\ell}(L) \leq \tau\ . \]

The following two Lemmas imply the above.
\begin{lemma}
$a^{\calL}_{L}(\ell) \leq \frac{\tau}{4}$.
\end{lemma}
\begin{proof}
To see this, note the selection  of $\tilde Q$ in the first slot of
each slot-pair. We claim that during this slot, all links in $L$ are
transmitting with linear power. For links in $L$ that were selected in
an earlier phase, this is obviously true. For links in $Q$ that will
be selected in $L$, this is true as well, since eventual admission in
$L$ is only possible (though not guaranteed) if the link decided to transmit during the first slot.

The proof of the Lemma is completed by noting the SINR threshold used in  the selection of $\tilde Q$.
\end{proof}

\begin{lemma}
$a^{\calU}_{\ell}(L) \leq \frac{\tau}{4}$. 
\end{lemma}
\begin{proof}
The selection process implemented during the second slot guarantees
that  $a^{\calL}_{L_d}(\ell_d) \leq \frac{\gamma_2 \tau}{4}$, where $L_d$ is the dual set of $L$ and $\ell_d$ is the dual of $\ell$ (this follows the proof of the previous Lemma almost verbatim).

To complete the proof, 
we use a result from \cite[Obs.~4]{KV10}. It was shown that for a constant $\gamma_2$, and links $\ell$ and $\ell'$,
\begin{claim}
$\gamma_2 a^{\calL}_{\ell'_d}(\ell_d) \leq a^{\calU}_{\ell}(\ell') \leq \frac1{\gamma_2} a^{\calL}_{\ell'_d}(\ell_d)$.
\label{dualityrelation}
\end{claim}

Using this claim, we get that
\[ a^{\calU}_{\ell}(L) = \sum_{\ell' \in L} a^{\calU}_{\ell}(\ell') 
   \leq \sum_{\ell' \in L_d} \frac1{\gamma_2} a^{\calL}_{\ell'_d}(\ell_d)
   = \frac1{\gamma_2} a^{\calL}_{L_d}(\ell_d) 
   \le \frac{\tau}{4}\ , \]
%
as required.
\end{proof}

\subsubsection{$\calT'$ is large}

Define, following \cite{SODA12,KesselheimSoda11},
\begin{equation*}
f_{\ell}(\ell') = \left\{
\begin{array}{rl}
a^{\calU}_{\ell'}(\ell) + a^{\calL}_{\ell}(\ell') & \text{if }  \ell \leq \ell',\\
0 & \text{otherwise}.\\
\end{array} \right.
\label{eq:amenability}
\end{equation*}
This definition is essentially equivalent to the definition of that of $f_{\ell}(\ell')$ of \cite{SODA12}  and of $w(\ell, \ell')$ of \cite{KesselheimSoda11} (also see Appendix \ref{sec:amenabilityprimer}). Those definitions are presented in terms
of distances.
The reason why we choose to define 
$f_{\ell}(\ell')$ in terms of affectances here, instead of distances, is that 
affectances (or their SINR equivalents) can be measured by the link receivers and thus used as a selection criteria.
For a set $X$, define $f_{\ell}(X) = \sum_{\ell' \in X}f_{\ell}(\ell')$
and $f_{X}(\ell') = \sum_{\ell \in X}f_{\ell}(\ell')$.

Recall that the input set $\calT$ is $O(1)$-sparse, which is of
crucial importance.
Consider once again the execution of the algorithm for phase $i$.
Let $\calT'_{i-1}$ be the selected set at the end of phase $i-1$. 
As before, let $Q$ be the links considered in phase $i$
and $Q^*$ be the links that succeeded in that phase. 
Since $\calT$ is $O(1)$-sparse, so is $Q$.

\begin{lemma}
 Let $Q'$ be the subset of links $\ell$ in $Q$ with
 $f_{\calT'}(\ell) \le \gamma_2^2 \cdot \tau/8$. Then, $\Ex(|Q^*|) = \Omega(|Q'|)$.
\label{samplingchoosesbig}
\end{lemma}
\begin{proof}
Consider any link $\ell \in Q'$. We shall show below that $\Pro(\ell
\in Q^*) = \Omega(1)$, which implies the Lemma.

In the first slot,
$\ell$ transmits with probability $p$. We claim that:
\begin{claim}
$\Pro(a^{\calL}_{T}(\ell) \leq \tau/8) \geq \frac12$, where $T$ are the links in $Q$ transmitting.
\label{slot1success}
\end{claim}
\begin{proof}
Let $\rho$ be such that length class in phase $i$ covers lengths in $[\rho, 2 \rho)$.
Since $Q$ is $O(1)$-sparse, it follows that balls of radius $\rho$ 
contain only a constant number of nodes that have links in $Q$. 
The claim now follows from arguments essentially identical to those in Lemma \ref{smalllinksuccess3}, after setting the probability $p$ sufficiently small.
\end{proof}

Since $\ell \in Q'$, we see that $a^{\calL}_{\calT'_{i-1}}(\ell)
\leq \tau/8$, by the definition of $Q'$.
Thus, if $a^{\calL}_{T}(\ell) \leq \tau/8$, 
then $a^{\calL}_{T \cup \calT'_{i-1}}(\ell) \leq \tau/4$, 
and the transmission is recorded as a success.
Thus, $\ell$ transmits and is recorded as a success with probability $\frac12 p$.
In other words, 
\begin{equation}
\Pro(\ell \in \tilde Q) \geq \frac12 p\ .
\label{qtildeprob}
\end{equation}

Now, condition on $\ell$ being in $\tilde Q$. Then $\ell_d$ transmits with probability $\gamma_2 p$. 
The following claim can be proven using Claim \ref{dualityrelation} and is similar to Claim \ref{slot1success}.
\begin{claim}
$\Pro(a^{\calL}_{T_d}(\ell_d) \leq \frac{\gamma_2 \cdot \tau}{8}) \geq \frac12$, where $T_d \subseteq \tilde Q_d$ are the (dual) links transmitting in this slot.
\end{claim}

 Following a argument similar to the one used for the first slot, we see that in the second slot, such a transmission is recorded as a success as well.

Thus, $\Pro(\ell \in Q^*|\ell \in \tilde Q) \geq \frac12 \gamma_2 p
$. Combining this with Eqn.~\ref{qtildeprob}, we get $\Pro(\ell \in
Q^*) \geq \frac14 \gamma_2 p^2 = \Omega(1)$, completing the proof of the Lemma.
\end{proof}

This leads to the desired bound on the size of $\calT'$.

\begin{theorem}
The set $\calT'$ chosen by the algorithm satisfies $\Ex(|\calT'|) = \Omega(|\calT(M)|)$.
\end{theorem}
\begin{proof}

By Thm.~\ref{sparsecapacity}, there exists a set $O \subseteq \calT$ such that
$O$ is feasible and $|O| = \Omega(|\calT|)$. Thus, it suffices to show that
$\Ex(|\calT'|) = \Omega(|O|)$ for any feasible set $O$.

Thm.~1 of \cite{KesselheimSoda11} shows that for a feasible link set $R$ and any link $\ell$,
\begin{equation}
f_{\ell}(R) =  O(1)\ .
\label{flbound}
\end{equation}

Consider the set $R = O \setminus \calT'$. We divide $R$ further into two
subsets: $R_1 = \{\ell' \in R: f_{\calT'}(\ell') > \gamma_2^2
\tau/8 \}$ and $R_2 = R \setminus R_1$.
Summing Eqn.~\ref{flbound} for all $\ell \in \calT'$,
\begin{equation}
f_{\calT'}(R) =  O(|\calT'|)\ .
\label{alglarge}
\end{equation}
By definition of $R_1$, $f_{\calT'}(R_1) > \frac18 \gamma_2^2 \tau |R_1|$. Assume first that $|R_1| \geq |R|/2$. Then, we get, $f_{\calT'}(R_1) > \frac1{16} \gamma_2^2 \tau |R|$, which combined with Eqn.~\ref{alglarge} gives, $|\calT'| = \Omega(f_{\calT'}(R)) \geq f_{\calT'}(R_1) = \Omega(|R|)$. Since $|O| \leq |\calT'| + |R|$, this clearly implies that $|\calT'| = \Omega(|O|)$. Otherwise assume, $|R_1| < |R|/2$ and thus $|R_2| > |R|/2$. But Lemma \ref{samplingchoosesbig} implies that $\Omega(|R_2|)$ links were chosen by the algorithm (in expectation), from which $\Ex(|\calT'|) = \Omega(|O|)$ follows.
\end{proof}




\subsubsection{Computing the power assignment}
So far we have dealt with the \emph{selection} of a large set of feasible links.
Once the link set $\calT'$ is identified, we must select the power assignments for 
this set. Given a set of links that are known to be feasible, there exists
a large body of work proposing algorithms that converge to a power assignment
making the assignment feasible. For example, two recent ones are \cite{DBLP:conf/infocom/LotkerPPP11} and \cite{DBLP:conf/icalp/DamsHK11}. Using such an algorithm as a black box, we can find the appropriate power assignment.

\begin{theorem}
  There exists a distributed algorithm that connects the nodes in
  $O(\log n)$ slots. Assuming that there exists an algorithm to find the power
  assignment for a feasible set in time $\eta$,
  this algorithm completes in time $O(\log  n (\log \Delta \cdot \log n + \eta))$.
\end{theorem}
As an example, if we select the algorithm from \cite{DBLP:conf/infocom/LotkerPPP11}, $\eta$ can be bounded by $O(\log \Delta (\log n + \log \log \Delta))$. This proves the first part of Thm.~\ref{thm:3}.

\section{Conclusions} 
Our distributed algorithms have efficiency and effectiveness that
appear to be close to best possible. An interesting direction would be
to treat dynamic situations, including asynchronous node wakeup, node
and link failures, and mobility.

\bibliographystyle{abbrv}
\bibliography{references}		

\appendix

\section{Missing Proofs}
\label{sec:missing}
\noindent \textbf{Proof of Thm.~\ref{constantsparsity}}
\begin{proof}
Recall that $M$ is the set of nodes of degree at most $\rho = \frac{160}{p^2}$
in $\calT$. For sets $X$ and $Y$, let $\edges(X, Y)$ be the number
of links with senders in $X$ and receivers in $Y$. We claim that setting $\calT(M) = \edges(M, M)$ fulfills the properties claimed in the theorem.

The $O(1)$-sparsity follows by noting that the nodes in $M$ have
degree $O(1)$; the proof of Lemma \ref{pointsinball} can be followed
verbatim using the constant-degree bound instead of the $O(\log n)$-bound employed there.

Thus, what remains to be proven is that $\Ex(|\edges(M, M)|) = \Omega(n) = \Omega(|\calT|)$.
Let $M' = P \setminus M$ (recall that $P$ is the set of all nodes). 
Since $\calT$ is a tree, $|\calT| = n - 1$. Then, since the number of unique links adjacent to $M$ is
at least $\frac12 M \rho$, it is easily
computed that $|M'| \leq \frac{2n}{\rho}$ and thus $|M| \geq n (1 - \frac2{\rho})$.
We show in Lemma \ref{mprimedegbound} below that $\Ex(|\edges(M', P)|) \leq \frac{n}{e^{9}}$. Note that
since $\calT$ is a connected tree, $|\edges(M, P)| \geq |M| - 1$. Thus, 
\begin{align*}
\Ex(|\edges(M, M)|) & \geq \Ex(|\edges(M, P)|) - \Ex(|\edges(M, M')|) \\
  &       \geq \Ex(|M|) - 1 - \Ex(|\edges(M', P)|) \\
  & \geq n \left(1 - \frac2{\rho}\right) - \frac{n}{e^{9}} 
  = \Omega(n)\ ,
\end{align*}
which implies the theorem.
\end{proof}

\begin{lemma}
$\Ex(|\edges(M', P)|) \leq \frac{n}{e^{9}}$.
\label{mprimedegbound}
\end{lemma}
\begin{proof}

Recall that by Thm.\ \ref{seconddegbound}, 
$\Pro(deg(u) \geq d) \leq e^{\frac{-p^2 d}{8}}$, where $deg(u)$ is the degree of $u$. 
This implies that $\Pro(deg(u) \in [d, 2d)) \leq e^{\frac{-p^2 d}{8}}$.
Since $\rho = \frac{160}{p^2}$, we can verify using basic calculus that 
$e^{p^2 \rho 2^t/8} \geq \rho^2 2^{2t+2}$, 
for all $t$. Using this bound, we get,
\begin{align*}
& \Ex(|\edges(M', P)|)  \leq n \sum_{t = 0}^{\infty} \Pro(deg(u) \in [\rho 2^t, \rho 2^{t+1})) \rho 2^{t+1} \\
  & \leq n \sum_{t = 0}^{\infty} e^{\frac{-p^2 \rho 2^t}{8}} \rho 2^{t+1}
    \leq n \sum_{t = 0}^{\infty} e^{\frac{-p^2 \rho 2^t}{16}}
  \leq n \sum_{t = 0}^{\infty} e^{-10\cdot 2^t} \\
&  = \frac{n}{e^{10}} + n \sum_{t = 1}^{\infty} \frac1{e^{10\cdot 2^t}} 
   \leq \frac{n}{e^{10}} + \frac{n}{e^{10}} \sum_{t = 1}^{\infty} \frac1{e^{2^t}} \leq \frac{2n}{e^{10}} \leq \frac{n}{e^{9}} \ .
\end{align*}
\end{proof}

\noindent \textbf{Proof of Lemma~\ref{averageaffectancemean}}

\begin{proof}
The proof of this Lemma follows ideas from \cite{us:esa09full} and \cite{SODA11}.
We need to relate the idea of sparsity to the idea of ``independence'' used in \cite{us:esa09full}.

We say that a set of links is \emph{$q$-independent} if any two of them, $\ell = (x, y)$ and $\ell'= (x', y')$, satisfy the
constraint $d(x, y') \cdot d(y, x') \geq q^2 d(x, y) \cdot d(x', y')$.

We claim, 
\begin{claim}
Let $C$ be a sufficiently large constant.
Let $Q$ be a $C$-independent set, and
for any link $\ell$ in $\calT'$, let $Q^\ell$ be the links in $Q$ longer than $\ell$.
Then, $a_\ell(Q^\ell) + a_{Q^\ell}(\ell) = O(\Upsilon)$.
\end{claim}
\begin{proof}
Partition $Q^\ell$ into two sets: $Q^\ell_l$, with links length at
least $d(x, y) \cdot 2 (2 \beta n)^{2/\alpha}$, and $Q^\ell_s$, with
the remaining links. It follows from \cite[Lemma 4.4]{us:esa09full}
that $a^\calM_{Q^\ell_l}(\ell) + a^\calM_\ell(Q^\ell_l) = O(\log \log
\Delta)$. On the other hand, $Q^\ell_s$ can be partitioned into
$O(\log n)$ length classes. For such sets, it is known
\cite{us:esa09full} that $C$-independence, for some constant $C$,
implies feasibility. Let $Z$ be such a set. By Lemma 7 of \cite{KV10},
$a^\calM_{Z}(\ell) = O(1)$. Since $Z$ belongs to a single length
class, it is also possible to show (following arguments similar to
\cite{KV10}) that $a^\calM_\ell(Z) = O(1)$.  Thus,
$a^\calM_{Q^\ell_s}(\ell) + a^\calM_\ell(Q^\ell_s) = O(\log n)$,
summing over the $O(\log n)$ such $Z$'s. The claim follows.
\end{proof}

By Lemma \ref{sparsitytoindependence} below, we know that $\calT'$ can
partitioned into a constant number of $C$-independent sets. 
Let $Q_1, Q_2, \ldots, Q_t$ be a partition of $L'$ into $t$ different $C$-independent sets.
For a link $\ell$, let $Q^\ell_i = \{\ell' \in Q_i : \ell' \ge \ell\}$.
Then, 
\begin{align*}
a^\calM_{\calT'}(\calT') 
   & \leq \sum_{\ell=(x, y) \in \calT'} \sum_{i=1}^t a^\calM_{\ell}(Q^{\ell}_i) + a^\calM_{Q^{\ell}_i}(\ell)\\
   & = t |\calT'| O(\Upsilon) = O(|\calT'| \Upsilon)\ ,
\end{align*}
since $t = O(1)$.
%
\end{proof}

\begin{lemma}
$\calT'$ can be partitioned into a constant number of $C$-independent sets.
\label{sparsitytoindependence}
\end{lemma}
\begin{proof}
Consider any link $\ell$. We claim that there are $O(1)$ links $\ell'$ at least as long as $\ell$ such that $\ell$ and $\ell'$ are \emph{not} $C$-independent. This claim proves the lemma by the following algorithm. 
Sort the links in an ascending order of their length, breaking ties arbitrarily.

Now consider the graph on links where there is an edge between links if they are not $C$-independent. 

By the claim, all links have $O(1)$ edges to links after them in the ascending order. Such a graph is $O(1)$-colorable, where each color represents an independent set in graph theoretic sense, and thus a $C$-independent set according to our definition.

Now we prove the claim. Recall that $\calT'$ is $\gamma_3$-sparse for some constant $\gamma_3$. Consider the link $\ell = (u, v)$ and a ball of radius $(2C)^2 \cdot d(u, v)$ around $u$. By a basic geometric
argument, this ball can be covered by $O(1)$ balls of radius $d(u, v)/8$. By the definition of sparsity, there can be at most $\gamma_3$ links of length $d(u, v)$ or higher that have one endpoint in each of the smaller balls. Thus, the larger ball also
contains only $O(1)$ such links. We now claim that all other links, i.e.,  $\ell' = (u', v')$ such that $\min(d(u', u),d(v',u)) \geq (2C)^2 \cdot d(u, v)$  are such that $\ell$ and $\ell'$ are $C$-independent.  
First, assume that $d(u', v) \geq \frac14 d(u', v')$. Then $d(u', v) \cdot d(u, v') \geq \frac14 d(u', v') \cdot (2C)^2 \cdot d(u, v) = C^2 d(u', v') d(u, v)$ which implies $C$-independence. On the other hand, if 
$d(u', v) < \frac14 d(u', v')$, then 
$d(u, v') \geq d(u', v') - d(u', v) - d(u, v) \geq d(u', v') - \frac{5}{4} d(u', v) \geq \frac{11}{16} d(u', v')$, from which $C$-independence follows by similar computations.
\end{proof}

\section{A Short Primer on Sparsity,\\ Amenability and Feasibility}
\label{sec:amenabilityprimer}


In \cite{SODA12}, 
a set of links $L$ was defined to be $\eta$-\textbf{amenable} if the following holds:
for any link $\ell$ ($\ell$ not necessarily a member of $L$), 
$\sum_{\ell' \in L} f_{\ell}(\ell') \leq \eta$,
for a function $f$ (see Eqn.~\ref{eq:amenability}), for some $\eta$. Actually, in \cite{SODA12},  $\eta$ is implicitly considered to be a constant, and just the term \textbf{amenable} is used. The definition extends naturally to arbitrary $\eta$.
It was shown in \cite{KesselheimSoda11} that an 
$\eta$-\textbf{amenable} set $L$ has a feasible subset of size $\Omega\left(\frac1{\eta} |L|\right)$.

Now the final ingredient needed is to tie sparsity to feasibility (and
thus get Thm.~\ref{sparsecapacity}). We claim that sparsity as defined
in this paper implies amenability. This is implicit in
\cite{SODA12}. Specifically, in proving the main Lemma 4 of \cite{SODA12}, it is first shown that the structure in question (which happens to be a Minimum Spanning Tree on the set of nodes) is $O(1)$-sparse (Lemma 5) and then this is used to show that the structure is amenable (which then implies a large feasible subset).

\section{A Note on the Approximation Factor for Distributed Scheduling}
\label{sec:dualclarification}


If acknowledgments have to be explicitly
implemented, the algorithm of \cite{KV10,icalp11} produces a schedule
length of $O((T + T') \cdot \log n)$, where $T$ is the optimal
schedule for the input link set, and $T'$ is the optimal schedule for
the dual set of the input set, which may be larger than $O(T \log  n)$. 
For our instance, this problem simply
is not relevant. The constructed link set
$\calT$ \emph{is its own dual}, and thus a $O(\log n)$-approximation
factor can be safely asserted.

\end{document}